\documentclass[a4paper, 11pt]{article}

\usepackage[T1]{fontenc}
\usepackage{amsmath,amsthm,mathtools}
\usepackage{libertine}
\usepackage[libertine]{newtxmath}
\usepackage{fullpage}
\usepackage{booktabs}
\usepackage{microtype}
\usepackage{url}
\usepackage{tikz}
\usetikzlibrary{decorations.pathreplacing,positioning}
\usepackage[hidelinks]{hyperref}
\hypersetup{
  pdftitle={Optimal Rank Lotteries for Truthful Unit-Interval Covering},
  pdfauthor={Alexandros A. Voudouris},
  pdfsubject={Truthful interval covering and randomized approximation mechanisms}
}

\newtheorem{theorem}{Theorem}
\newtheorem{lemma}{Lemma}

\newtheorem{claim}{Claim}
\theoremstyle{remark}

\theoremstyle{definition}
\newtheorem{example}{Example}

\DeclareMathOperator{\SC}{SC}
\newcommand{\E}{\mathbb{E}}
\newcommand{\R}{\mathbb{R}}
\newcommand{\hats}{\hat s}
\newcommand{\bfs}{\mathbf{s}}
\newcommand{\bfhats}{\mathbf{\hats}}

\begin{document}

\title{\bf Optimal Rank Lotteries for Truthful Unit-Interval Covering}
\author{Alexandros A. Voudouris}

\date{
Department of Mathematics and Computer Science \\
University of Southern Denmark, Denmark}

\maketitle

\begin{abstract}
In truthful interval covering, each agent has a private interval of unit length, and the goal is to decide where to place a public unit interval so as to minimize the total uncovered length while incentivizing the agents to be truthful. Previous work proposed a universally strategyproof lottery over order-statistic mechanisms with approximation ratio at most $5/3$, and established a lower bound of $3/2-o(1)$ for this class of mechanisms. We close this gap by characterizing the approximation ratio of every report-independent lottery over order-statistics. In particular, we show that, for every $n\ge2$, the optimal approximation ratio for this class is $\frac32-\frac{1}{2\lfloor n/2\rfloor}$. Beyond rank lotteries, we show that every universally strategyproof randomized mechanism has approximation ratio at least $9/8$.
\end{abstract}

\medskip
\noindent\textbf{Keywords:} mechanism design without money; strategyproofness; approximation ratio; order statistic; interval covering

\section{Introduction}\label{sec:intro}
The truthful interval covering (TIC) problem was introduced by Deligkas et al.~\cite{deligkas2024} to model settings in which a unit-length public interval must be positioned to maximize its overlap with the privately known unit intervals of a set of agents, while incentivizing them to be truthful. A natural example is scheduling a public service that is available for a fixed period of time, when each agent has a specific availability window and benefits from the extent to which it overlaps with the chosen service window. TIC is an instance of {\em approximate mechanism design without money}, a framework introduced by Procaccia and Tennenholtz~\cite{procaccia2013} through strategic facility location problems; see the survey of Chan et al.~\cite{chan2021}.

Among other results, Deligkas et al.~\cite{deligkas2024} established a tight approximation ratio of $2-2/n$ for deterministic strategyproof mechanisms, proposed a universally strategyproof randomized mechanism, called {\sc Uniform-Statistic}, with approximation ratio at most $5/3$, and showed a lower bound of $3/2-o(1)$ for randomized mechanisms that are convex combinations of order-statistic rules. Such rules form a classical family of strategyproof mechanisms for single-peaked domains~\cite{moulin1980}. In this paper, we first determine how to optimally randomize over order-statistic rules with respect to social cost. We then move beyond rank lotteries and establish the first lower bound for general universally strategyproof randomized mechanisms.

Our first contribution is an exact characterization of the worst-case approximation ratio of every report-independent lottery over order statistics. The characterization depends only on the probability assigned to consecutive blocks of ranks, with each block value attained by a canonical profile of a simple form. The key geometric ingredient is a random shifted-grid representation of the truncated line metric, which shows that, relative to any fixed benchmark location, there exists an extremal profile in which one consecutive block of agents is located at the benchmark and all remaining agents are mutually isolated. This reduces the continuous geometric problem to an optimization over consecutive rank blocks. As a first application, we determine the exact finite-$n$ approximation ratio of the {\sc Uniform-Statistic} mechanism of Deligkas et al.~\cite{deligkas2024}, refining its asymptotically tight $5/3$ guarantee.

We then optimize the block formula over all rank distributions. For $n=2$, every rank lottery is socially optimal, whereas for $n=3$ the deterministic median rule is the unique optimal rank lottery. For every $n\ge4$, we show that the best achievable approximation ratio is
\[
\frac32-\frac{1}{2\lfloor n/2\rfloor},
\]
and give explicit central lotteries attaining this value. This matches, for every $n\ge4$, the finite-$n$ lower bound obtained from the construction underlying the $3/2-o(1)$ bound from previous work.

Finally, we move beyond rank lotteries and establish, using a simple four-agent construction, the first general lower bound of $9/8$ for universally strategyproof randomized mechanisms. Our arguments rely only on the inequalities implied by strategyproofness, without assuming anonymity or rank-basedness.

\section{Model and rank-based mechanisms}\label{sec:model}
There are $n\ge 2$ agents. Each agent $i$ has a private unit interval $I_i=[s_i,s_i+1]$. Since the length of all agent intervals is $1$, we can identify each agent by her {\em starting} location $s_i$; let $\bfs = (s_i)_i$ be the {\em location profile} of all agents.  If a {\em covering} unit interval starts at location $y$, agent $i$ incurs {\em cost}
\[
d(s_i,y)=1-|I_i\cap[y,y+1]| = \min\{1,|s_i-y|\}.
\]
The function $d$ is the usual line metric truncated at $1$ and is itself a metric. In particular, its triangle inequality follows from the fact that $\min\{1,a+b\}\le \min\{1,a\}+\min\{1,b\}$ for any $a,b \ge 0$.
The social cost of a covering location $y$ for profile $\bfs$ is 
\[
\SC(y;\bfs)=\sum_{i=1}^n d(s_i,y).
\]
We denote by $\operatorname{OPT}(\bfs)=\min_{y\in\R}\SC(y;\bfs)$ the {\em optimal} social cost achieved by any covering location for $\bfs$. 

A deterministic mechanism $f$ takes as input a location profile $\bfs$ and outputs a covering location $y = f(\bfs)$. It is  {\em strategyproof} if, for every agent $i$, true location $s_i$, alternative report $z_i$, and reports $\bfs_{-i}$ of
the other agents,
\[
d\bigl(s_i,f(s_i,\bfs_{-i})\bigr) \le d\bigl(s_i,f(z_i,\bfs_{-i})\bigr).
\]
A randomized mechanism $M$ outputs a probability distribution over all possible covering locations. It is \emph{universally strategyproof} if it is a report-independent lottery over deterministic strategyproof mechanisms.
Its approximation ratio is
\[
\rho(M)=\sup_{\bfs:\operatorname{OPT}(\bfs)>0}
\frac{\E[\SC(M(\bfs);\bfs)]}{\operatorname{OPT}(\bfs)}.
\]
If $\operatorname{OPT}(\bfs)=0$, all agent locations coincide and choosing any of them gives $0$ social cost. 

\paragraph{Rank lotteries}
We focus on a class of randomized mechanisms known as {\em rank lotteries}, or {\em mixtures of order-statistic mechanisms}, which are universally strategyproof~\cite{deligkas2024}. Let $s_1\le\ldots\le s_n$ be the {\em sorted} reported agent starting locations, breaking possible ties in an arbitrary but fixed way. Indices henceforth denote rank positions. Tied ranks have the same coordinate, and thus their ordering does not affect the output or the social cost.

Let $p=(p_1,\ldots,p_n)$ be a probability vector. Independently of the reports, the rank lottery mechanism $M_p$ selects rank $r$ with probability $p_r$ and outputs $s_r$. We will write 
\[
\SC_p(\bfs)=\E[\SC(M_p(\bfs);\bfs)]
\]
for the expected social cost of $M_p$ when given $\bfs$. The expected social cost also has a useful pairwise form:
\begin{align}
\SC_p(\bfs)
&=\sum_{r=1}^n p_r\sum_{i=1}^n d(s_i,s_r)\notag\\
&=\sum_{1\le i<j\le n}(p_i+p_j)\cdot d(s_i,s_j).
\label{eq:pairwise}
\end{align}
Indeed, for a fixed unordered pair of ranks $\{i,j\}$, the distance $d(s_i,s_j)$ appears only once when rank $i$ is selected and once when rank $j$ is selected. In particular, every pairwise coefficient is non-negative. This monotonicity is the basis of the extremal reduction presented in the next section. 

\section{Exact approximation formula}\label{sec:formula}
The main step towards deriving the exact approximation ratio is an extremal reduction. Given an arbitrary profile and a fixed benchmark location, we show that it suffices to consider profiles in which some consecutive block of agents is located at the benchmark, while all remaining agents are mutually isolated and at maximum truncated
distance from the benchmark. For a benchmark location $o$ and a scalar $\lambda$, we will show through the following existential lemma that the \emph{gap}
\[
\SC_p(\bfs)-\lambda\cdot \SC(o;\bfs)
\]
cannot decrease under the above reduction. Its proof is based on a simple application of the probabilistic method: a random shifted-grid rounding preserves the expected gap, and therefore there exists a deterministic realization with no smaller gap. A final deterministic separation step then transforms this realization into the desired canonical profile without decreasing the gap.

\begin{lemma}[Extremal reduction]\label{lem:reduction}
For every ordered profile $\bfs$, benchmark location $o$, and scalar $\lambda$,
there is an ordered profile $\bfhats$ with
\[
\SC_p(\bfs)-\lambda\cdot \SC(o;\bfs)
\le
\SC_p(\bfhats)-\lambda\cdot \SC(o;\bfhats)
\]
such that the agents at $o$ form one consecutive rank block, possibly empty, and every other agent is at truncated distance $1$ from $o$ and from every other non-central agent. Throughout the transformation, each agent retains the same global rank.
\end{lemma}

\begin{proof}
To simplify our discussion, we first translate all coordinates so that $o=0$; clearly, this changes no cost. 
We now draw $U$ uniformly from $[0,1)$ and, for each coordinate $x$, define
\[
R_U(x)=\lfloor x+U\rfloor.
\]
The map $R_U$ is non-decreasing, and thus applying it coordinate-wise preserves the global rank order. If rounding creates a tie, we retain the inherited order within that tie; because all tied reports have the same coordinate, this convention does not change the mechanism's output distribution.

Now observe that, for every $x,y\in\R$,
\begin{equation}\label{eq:grid-identity}
\E_U\bigl[d(R_U(x),R_U(y))\bigr]
=d(x,y).
\end{equation}
Indeed:
\begin{itemize}
\item If $|x-y|\ge1$, the two rounded values are always distinct and
both sides of~\eqref{eq:grid-identity} are equal to $1$. 
\item If $|x-y|<1$, suppose that $x<y$. 
The rounded values differ exactly when the interval $(x+U,y+U]$ contains an integer boundary. Equivalently,
writing $\{x+U\}$ for the fractional part of $x+U$, this happens when
$\{x+U\}\ge 1-(y-x)$.
Since $\{x+U\}$ is uniformly distributed on $[0,1)$, this event has probability $|x-y|$.
\end{itemize}
Since $R_U(0)=0$, Equation~\eqref{eq:grid-identity} also gives
\[
\E_U[d(0,R_U(x))]=d(0,x).
\]
Let $R_U(\bfs)$ denote the coordinate-wise rounded profile. 
Using \eqref{eq:grid-identity} and the pairwise representation~\eqref{eq:pairwise}, we have
\begin{align*}
\E_U\bigl[
\SC_p(R_U(\bfs))-\lambda\cdot \SC(0;R_U(\bfs))
\bigr]
= \SC_p(\bfs)-\lambda\cdot \SC(0;\bfs).
\end{align*}
Consequently, there exists a realization $u$ of $U$ whose gap is at
least that of the original profile.

The agents rounded to $0$ form a consecutive rank block. Every other rounded agent lies at a non-zero integer and hence has benchmark cost $1$. Keeping their order fixed, move the non-zero prefix and suffix to mutually distinct coordinates that are at distance at least $1$ from $0$ and from one another. Every truncated pairwise distance involving a non-central agent is then $1$. Thus no pairwise distance decreases, and~\eqref{eq:pairwise} implies that $\SC_p$ cannot decrease, while the
benchmark cost is unchanged. Translating back to $o$ yields the required profile $\bfhats$.
\end{proof}

We refer to profiles having the structure described in Lemma~\ref{lem:reduction} as \emph{canonical profiles}; see Figure~\ref{fig:canonical}. For a reduced profile of this form with a non-empty proper central block $B$, selecting a rank in $B$ yields social cost $n-|B|$, whereas selecting a rank outside $B$ yields social cost $n-1$. These two possible costs lead directly to the exact approximation formula below.

\begin{figure}[t]
\centering
\begin{tikzpicture}[x=0.63cm,y=0.76cm,font=\small]
  \draw[->] (0,1.25)--(8.4,1.25);
  \foreach \x in {0.5,1.5,6.8,7.8}
    \fill (\x,1.25) circle (1.5pt);
  \foreach \dy in {-0.12,-0.06,0,0.06,0.12}
    \fill (4.15,{1.25+\dy}) circle (1.35pt);
  \draw[dashed] (4.15,0.85)--(4.15,1.72) node[above] {$o$};
  \draw[decorate,decoration={brace,amplitude=4pt,mirror}]
    (3.45,0.74)--(4.85,0.74) node[midway,below=5pt] {$B$: central ranks};
  \draw[decorate,decoration={brace,amplitude=4pt}]
    (0.35,1.76)--(1.65,1.76) node[midway,above=5pt,font=\scriptsize] {isolated};
  \draw[decorate,decoration={brace,amplitude=4pt}]
    (6.65,1.76)--(7.95,1.76) node[midway,above=5pt,font=\scriptsize] {isolated};
\end{tikzpicture}
\caption{A canonical profile: one consecutive block $B$ lies at the benchmark $o$, and all remaining ranks are mutually isolated.}
\label{fig:canonical}
\end{figure}
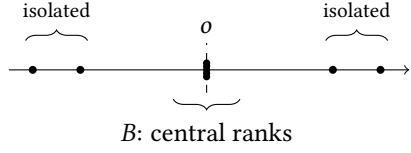

\begin{theorem}[Exact block formula]\label{thm:formula}
For any $n \geq 2$, let $\mathcal B_n$ denote the family of all non-empty proper consecutive
blocks of $\{1,\ldots,n\}$. For a rank distribution $p$ and
$B\in\mathcal B_n$, write $p(B)=\sum_{r\in B}p_r$
and define
\begin{equation}\label{eq:blockvalue}
\Phi_p(B) := 1+\frac{|B|-1}{n-|B|}(1-p(B)).
\end{equation}
Then
\begin{equation}\label{eq:formula}
\rho(M_p) = \max_{B\in\mathcal B_n}\Phi_p(B).
\end{equation}
Moreover, every value $\Phi_p(B)$ is attained by a canonical profile.
\end{theorem}

\begin{proof}
Let $\lambda:=\max_{B\in\mathcal B_n}\Phi_p(B).$
Since  $\Phi_p(B)=1$ for any singleton block $B$, we have that $\lambda\ge1$. 
Now, fix an arbitrary profile $\mathbf{s}$ and an arbitrary benchmark location $o$. 
We apply Lemma~\ref{lem:reduction} for this value of $\lambda$, and let $\mathbf{\hats}$ be the resulting reduced profile.

We first consider two extreme cases.
First, if the central block is empty, then by Lemma~\ref{lem:reduction} every agent is at truncated distance $1$ from $o$ and every pair of agents is at truncated distance $1$. Thus, $\SC(o;\mathbf{\hats})=n$. Whenever a rank is selected, the corresponding agent pays zero and all other $n-1$ agents pay one, leading to $\SC_p(\mathbf{\hats})=n-1$.
Since $\lambda \geq 1$, we obtain
\[
\SC_p(\mathbf{\hats})=n-1
\le n
\le \lambda \cdot n
=\lambda \cdot \SC(o;\mathbf{\hats}).
\]
Second, if the central block contains all agents, then both the benchmark cost and the expected cost of the mechanism are zero. Consequently, the inequality holds trivially.

It remains to consider a non-empty proper central block $B$. As observed above, the mechanism incurs cost $n-|B|$ when the selected rank lies in $B$, and cost $n-1$ otherwise. Therefore,
\begin{align}
\SC_p(\mathbf{\hats})
&=p(B)(n-|B|)+(1-p(B))(n-1)\notag\\
&=n-|B|+(1-p(B))(|B|-1).
\label{eq:blockcost}
\end{align}
Since $\SC(o;\mathbf{\hats}) = n-|B|$,  we have that 
\[
\frac{\SC_p(\mathbf{\hats})}{\SC(o;\mathbf{\hats})}
=
1+\frac{|B|-1}{n-|B|}(1-p(B))
=
\Phi_p(B)
\le\lambda.
\]
Thus, in every case, the reduced profile satisfies
\[
\SC_p(\mathbf{\hats})\le\lambda \cdot \SC(o;\mathbf{\hats}).
\]
The extremal reduction of Lemma~\ref{lem:reduction} then implies
\[
\SC_p(\mathbf{s})\le\lambda \cdot \SC(o;\mathbf{s})
\]
for the original profile $\mathbf{s}$. Since this inequality holds for every benchmark $o$, choosing
$ o\in\arg\min_{y\in\R}\SC(y;\mathbf{s})$
gives
$$
\SC_p(\mathbf{s})\le\lambda\cdot \operatorname{OPT}(\mathbf{s}).
$$
Taking the supremum over profiles with positive optimum proves the desired upper bound in~\eqref{eq:formula}.

\emph{Matching instance.}
We next present a profile $\bfs$ that attains this approximation ratio exactly. 
Fix $B\in\mathcal B_n$ and place exactly the ranks in $B$ at $0$. Place the ranks preceding $B$ to the left of $0$ and those following $B$ to the right, so that every non-central agent is at distance at least $2$ from $0$ and every pair of non-central agents is at distance at least $2$ from each other. For this canonical profile, location $0$ has cost $n-|B|$ and is optimal. Indeed, if the output location is $y$ such that $|y|<1$, then every isolated agent pays $1$, leading to a cost of at least $n-|B|$. Otherwise, if $|y|\ge1$, every central agent pays $1$ and, because the isolated agents are at pairwise distance $2$, at most one isolated agent can pay less than $1$. The cost is then at least
\[
|B|+(n-|B|-1)=n-1\ge n-|B|.
\]
Hence  $\operatorname{OPT}(\bfs) = n-|B|.$
The expected cost of the mechanism is exactly~\eqref{eq:blockcost}, and thus
\[
\frac{\SC_p(\bfs)}{\operatorname{OPT}(\bfs)}
=
1+\frac{|B|-1}{n-|B|}(1-p(B))
=
\Phi_p(B).
\]
Thus every block value is attained by a canonical profile. Maximizing over $B\in\mathcal B_n$ proves equality in~\eqref{eq:formula}.
\end{proof}


\section{Optimal rank lotteries}\label{sec:consequences}
We now use the characterization of the previous section to derive exact bounds for rank lotteries. We first illustrate the block formula by revisiting the {\sc Uniform-Statistic} mechanism of Deligkas et al.~\cite{deligkas2024}. For $n=6q$, our characterization gives its exact finite-$n$ approximation ratio, refining its asymptotically tight $5/3$ guarantee. We then determine the optimal rank lotteries, first for the special cases $n\in\{2,3\}$ and subsequently for every $n\ge4$.

\begin{example}[Recovery of the previous $5/3$ bound of \cite{deligkas2024}]\label{rem:uniform}
For simplicity, assume that $n=6q$ for some integer $q\ge1$. The {\em Uniform-Statistic} mechanism assigns probability $1/3$ to each of the three ranks $2q$, $3q$, and $4q$. Hence, for any consecutive block $B$, its probability mass $p(B)$ is determined entirely by how many of these three ranks it contains.

By Theorem~\ref{thm:formula}, the approximation ratio of
Uniform-Statistic is
\[
\rho(M_p)=\max_{B\in\mathcal B_n}\Phi_p(B), 
\]
where 
\[
\Phi_p(B) = 1+\frac{|B|-1}{6q-|B|}(1-p(B)).
\]
For a fixed value of $p(B)$, $\Phi_p(B)$ is increasing in $|B|$. Thus, for each possible value of $p(B)$, it suffices to determine the {\em largest consecutive block} having each possible probability mass and then compare the resulting values.
We distinguish four cases.
\begin{itemize}
\item Case 1: $B$ contains all three supported ranks. Then $p(B)=1$, and therefore
$\Phi_p(B)=1$. Clearly, such blocks cannot determine the worst-case approximation ratio.

\item Case 2: $B$ contains exactly two supported ranks. Since the
support is $\{2q,3q,4q\}$, the block must contain either $2q$ and $3q$
but not $4q$, or $3q$ and $4q$ but not $2q$. In the first case, the
largest possible block is $\{1,\ldots,4q-1\}$ which has size $4q-1$. 
In the second case, the largest possible block is $B^\star=\{2q+1,\ldots,6q\}$
which has size $4q=2n/3$. Thus, within this case, the maximum is attained by
$B^\star$. Since $p(B^\star)=2/3$, we obtain
\begin{align*}
\Phi_p(B^\star)
&= 1+\frac{4q-1}{2q}\left(1-\frac23\right)\\
&= \frac53-\frac{1}{6q}
= \frac53-\frac1n.
\end{align*}

\item Case 3: $B$ contains exactly one supported rank, so that
$p(B)=1/3$. If the unique supported rank is $2q$, the block can have
size at most $3q-1$; if it is $3q$, the size is at most $2q-1$; and if
it is $4q$, the size is at most $3q$. Hence every such block has size at most $3q=n/2$. Therefore,
\begin{align*}
\Phi_p(B)
&\le 1+\frac{3q-1}{3q}\left(1-\frac13\right)\\
&= \frac53-\frac{2}{9q}
< \frac53-\frac1n.
\end{align*}

\item Case 4: $B$ contains none of the supported ranks. Then $p(B)=0$. Such a block must lie entirely in one of the gaps determined
by $2q$, $3q$, and $4q$. The largest such gap is $\{4q+1,\ldots,6q\}$ of size $2q=n/3$. Consequently,
\begin{align*}
\Phi_p(B)
&\le 1+\frac{2q-1}{4q}\\
&= \frac32-\frac{1}{4q}
< \frac53-\frac1n.
\end{align*}
\end{itemize}
Overall, the largest block value is attained by $B^\star=\{2q+1,\ldots,6q\}$ and Theorem~\ref{thm:formula} implies that the exact approximation ratio of the mechanism is $5/3 - 1/n$.
\hfill $\qed$
\end{example}

We now turn from evaluating a fixed rank lottery to optimizing over all rank distributions. We first show that, for two and three agents, rank lotteries can achieve social optimality.

\begin{theorem}[Optimality for two and three agents]\label{thm:small-n}
For $n=2$, every rank lottery is socially optimal. For $n=3$, the unique optimal rank lottery is the deterministic median mechanism.
\end{theorem}

\begin{proof}
For $n=2$, the family $\mathcal B_2$ consists only of the two singleton
blocks. Since every singleton block $B$ satisfies
\[
\Phi_p(B)
=
1+\frac{|B|-1}{2-|B|}(1-p(B))
=1,
\]
Theorem~\ref{thm:formula} gives
\[
\rho(M_p)=1
\]
for every rank distribution $p$. Thus every rank lottery is optimal when there are two agents.

For $n=3$, the family $\mathcal B_3$ consists of the three singleton
blocks and the two consecutive blocks of size two. For every singleton
$B=\{r\}$, we have
\[ 
\Phi_p(\{r\}) = 1+\frac{1-1}{3-1}(1-p_r) =1.
\]
For the two blocks of size two, using $p_1+p_2+p_3=1$, we obtain
\begin{align*}
\Phi_p(\{1,2\})
&=
1+\frac{2-1}{3-2}\bigl(1-p_1-p_2\bigr)
=1+p_3,\\
\Phi_p(\{2,3\})
&=
1+\frac{2-1}{3-2}\bigl(1-p_2-p_3\bigr)
=1+p_1.
\end{align*}
Hence, Theorem~\ref{thm:formula} yields
\[
\rho(M_p)
=
\max\{1,1+p_1,1+p_3\}
=
1+\max\{p_1,p_3\}.
\]
Since $p_1,p_3\ge0$, this ratio is minimized at $1$, and equality holds if and only if $p_1=p_3=0$.
Because $p_1+p_2+p_3=1$, the unique optimal rank distribution is therefore $p=(0,1,0)$, which is the deterministic median mechanism.
\end{proof}

The situation changes for instances with at least four agents: social optimality can no longer be achieved by a rank lottery. We next determine the best possible approximation ratio for every $n\ge4$ and present rank lotteries that attain it.

\begin{table}[t]
\centering
\caption{An optimal central lottery for each residue class of $n \geq 4$.}
\label{tab:central-lotteries}
\small
\begin{tabular}{@{}ll@{}}
\toprule
Population size & Central lottery \\
\midrule
$n=4q$ & uniform on $q+1,\ldots,3q$ \\
$n=4q+1$ & equal mixture of the uniform lotteries on \\
& $q+1,\ldots,3q$ and $q+2,\ldots,3q+1$ \\
$n=4q+2$ & uniform on $q+2,\ldots,3q+1$ \\
$n=4q+3$ & uniform on $q+2,\ldots,3q+2$ \\
\bottomrule
\end{tabular}
\end{table}

\begin{theorem}[Optimal rank lotteries]\label{thm:optimal}
For every $n\ge4$,
\[
\min_p\rho(M_p)
=\frac32-\frac{1}{2\lfloor n/2\rfloor}.
\]
The lotteries in Table~\ref{tab:central-lotteries} attain this value.
\end{theorem}

\begin{proof}
The lower-bound construction of Deligkas et al.~\cite{deligkas2024}, with the analogous uneven split for odd $n$, yields an approximation ratio of at least
\[
\frac32-\frac{1}{2\lfloor n/2\rfloor}
\]
for every rank distribution $p$, once we account for the fact that, whenever an isolated agent is selected, that agent incurs zero cost. It therefore remains only to prove a matching upper bound.

We present the complete argument for $n=4q$, where $q\ge1$.
Consider the first distribution in Table~\ref{tab:central-lotteries}, namely
\[
p_r=
\begin{cases}
\dfrac{1}{2q}, & q+1\le r\le3q,\\[1mm]
0, & \text{otherwise}.
\end{cases}
\]
In words, the mechanism chooses uniformly among the middle $2q$ ranks.
By Theorem~\ref{thm:formula}, it suffices to show that
\[
\Phi_p(B)\le \frac32-\frac{1}{4q}
\]
for every non-empty proper consecutive block $B$.

Fix such a block $B$. Since
\[
\Phi_p(B)
=
1+\frac{|B|-1}{4q-|B|}(1-p(B)),
\]
among blocks of the same size, $\Phi_p(B)$ is maximized when $p(B)$ is minimized. Since $p$ is uniform on the consecutive ranks $\{q+1,\ldots,3q\}$ and zero elsewhere, the minimum probability mass among blocks of a given size is attained by placing the block as far as possible toward one of the two ends of the rank order. By symmetry, it therefore suffices to consider a prefix $B = \{1,\ldots,|B|\}$. We distinguish three cases according to how $B$ intersects the support
of $p$.

\begin{itemize}
\item {\bf Case 1: $|B|\le q$.}
Then $B$ contains no rank in the support of $p$, so $p(B)=0$. Hence
\[
\Phi_p(B) = 1+\frac{|B|-1}{4q-|B|}.
\]
This expression is increasing in $|B|$, and therefore, since $|B| \le q$, 
\[
\Phi_p(B) \le 1+\frac{q-1}{3q} = \frac43 - \frac{1}{3q} \leq \frac32 - \frac{1}{4q}. 
\]

\item {\bf Case 2: $q+1\le |B|\le 3q-1$.}
Then $B$ contains $|B|-q$ ranks in the support of $p$, and hence
\[
p(B)=\frac{|B|-q}{2q}.
\]
Substituting in the definition of $\Phi_p(B)$, we get
\[
\Phi_p(B) = 1+\frac{(|B|-1)(3q-|B|)}{2q(4q-|B|)}.
\]
To show that this is at most $3/2-1/(4q)$, it suffices to verify the inequality
\[ 2(|B|-1)(3q-|B|) \le (2q-1)(4q-|B|).
\]
The difference between the right- and left-hand sides is
\[
(|B|-2q)(2|B|-4q-1),
\]
which is non-negative because $|B|$ is an integer. Indeed, if $|B|\le2q$, both factors are non-positive, whereas if $|B|\ge2q+1$, both of them are positive. Hence, we obtain
\[
\Phi_p(B)\le\frac32-\frac{1}{4q}.
\]

\item {\bf Case 3: $3q\le |B|\le4q-1$.}
Then $B$ contains the entire support of $p$, and so $p(B)=1$, which implies that
\[
\Phi_p(B)=1 \le \frac32-\frac{1}{4q}.
\]
\end{itemize}
We have now established that every non-empty proper consecutive block satisfies
\[
\Phi_p(B)\le\frac32-\frac{1}{4q}.
\]
Moreover, equality is attained by the prefix $B=\{1,\ldots,2q\}$, for which $p(B)=1/2$. Indeed,
\begin{align*}
\Phi_p(B) &= 1+\frac{2q-1}{2q}\left(1-\frac12\right) = \frac32-\frac{1}{4q}.
\end{align*}
Therefore, Theorem~\ref{thm:formula} gives
\[
\rho(M_p)=\frac32-\frac{1}{4q},
\]
which completes the proof for $n=4q$.

The cases $n\in\{4q+1,4q+2,4q+3\}$ follow analogously from the corresponding distributions in Table~\ref{tab:central-lotteries}. Each of these distributions is symmetric, with probabilities weakly increasing from either edge toward the center. Hence, among consecutive blocks of a fixed size, an edge block minimizes $p(B)$ and therefore maximizes $\Phi_p(B)$. Repeating the same three-case calculation gives
\[
\rho(M_p) = \frac32-\frac{1}{2\lfloor n/2\rfloor},
\]
as desired.
\end{proof}

\section{Beyond rank lotteries}\label{sec:beyond}
We now move from rank lotteries and show a lower bound of $9/8$ for any universally strategyproof randomized mechanism.

\begin{theorem}[Unrestricted universal lower bound for four agents]
\label{thm:universal-lb}
For $n=4$, the approximation ratio of any universally strategyproof randomized mechanism is at least $9/8$.
\end{theorem}

\begin{proof}
Consider the following three location profiles: 
\[
A=\left(\tfrac12,\tfrac12,\tfrac32,2\right),\qquad
C=\left(0,\tfrac12,\tfrac32,2\right),\qquad
B=\left(0,\tfrac12,2,2\right).
\]
Before arguing about the performance of any mechanism when given these profiles as input, we make the following observation about the optimal social cost for $A$ and $B$. 

\begin{claim}\label{clm:endpoint-opt}
$\operatorname{OPT}(A)=\operatorname{OPT}(B)=2$.
\end{claim}

\begin{proof}
Let $y$ be an arbitrary location.
\begin{itemize}
\item 
For profile $A$, we can pair one agent at $1/2$ with the agent at $3/2$ and the other agent at $1/2$ with the agent at $2$. Using the triangle inequality, we have
\begin{align*}
\SC(y;A) 
&= \bigg(d(\tfrac12,y)+d(\tfrac32,y)\bigg) + \bigg(d(\tfrac12,y)+d(2,y)\bigg) \\
& \ge d(\tfrac12,\tfrac32)+d(\tfrac12,2) =2.
\end{align*}
Equality is attained at $y=1/2$, and thus $\operatorname{OPT}(A)=2$.

\item 
For profile $B$, we can pair one agent at $2$ with the agent at $0$, and the other agent at $2$ with the agent at $1/2$. 
Using the triangle inequality, we have
\begin{align*}
\SC(y;B)
&= \bigg(d(0,y)+d(2,y)\bigg) + \bigg(d(\tfrac12,y)+d(2,y)\bigg) \\
&\ge d(0,2)+d(\tfrac12,2)
=2.
\end{align*}
Equality is attained at $y=2$, and thus $\operatorname{OPT}(B)=2$.
\end{itemize}
The proof is complete. 
\end{proof}

Now let $f$ be an arbitrary deterministic strategyproof mechanism. Denote the locations that $f$ outputs for the three profiles by $y_A=f(A)$, $w=f(C)$, and $y_B=f(B)$. Profiles $A$ and $C$ differ only in the report of the first agent ($1/2$ in $A$, $0$ in $C$). Since $f$ is strategyproof, the cost of the first agent must not decrease when going from $A$ to $C$ or from $C$ to $A$. Therefore, the following two inequalities must hold:
\begin{equation}\label{eq:lb-left}
d\!\left(\tfrac12,y_A\right)\le d\!\left(\tfrac12,w\right)
\qquad
d(0,w)\le d(0,y_A).
\end{equation}
Similarly, profiles $C$ and $B$ differ only in the report of the third agent ($3/2$ in $C$, $2$ in $B$), and her cost must thus not decrease when going from $C$ to $B$ or from $B$ to $C$, leading to the following two inequalities:
\begin{equation}\label{eq:lb-right}
d\!\left(\tfrac32,w\right)\le d\!\left(\tfrac32,y_B\right)
\qquad
d(2,y_B)\le d(2,w).
\end{equation}
We will show that \eqref{eq:lb-left} and \eqref{eq:lb-right} imply that
\begin{equation}\label{eq:det-sum}
\SC(y_A;A)+\SC(y_B;B)\ge\frac92.
\end{equation}
Assuming \eqref{eq:det-sum} is indeed true, consider any universally strategyproof $M$. By definition, $M$ is a probability distribution over deterministic strategyproof mechanisms. Since \eqref{eq:det-sum} holds for every deterministic
strategyproof mechanism, it remains true after taking expectations, giving us that
\[
\E[\SC(M(A);A)] + \E[\SC(M(B);B)] \ge\frac92.
\]
So, the expected cost of $M$ must be at least $9/4$ for at least one of $A$ and $B$. Since $\operatorname{OPT}(A)=\operatorname{OPT}(B)=2$, this implies that the approximation ratio of $M$ is at least $9/8$. So, to complete the proof it remains to prove \eqref{eq:det-sum}. We do this by considering three ranges for $w$.

\medskip
\noindent 
\textbf{Case 1: $w\le1/2$.}
In this case, $d(\tfrac32,w)=1$. The first inequality in~\eqref{eq:lb-right} therefore implies $d(\tfrac32,y_B)=1$,
and hence
\[
y_B\le\frac12 \qquad\text{or}\qquad
y_B\ge\frac52.
\]
\begin{itemize}
\item 
If $y_B\le1/2$, the two agents at $2$ each incur
cost one, while the agents at $0$ and $1/2$ together incur cost at least
\[
d(0,y_B)+d(\tfrac12,y_B) \ge d(0,\tfrac12) =\frac12
\]
by the triangle inequality. Hence, overall
\[
\SC(y_B;B)\ge\frac52.
\]

\item 
If $y_B\ge5/2$, the agents at $0$ and $1/2$ each
incur cost one, while each of the two agents at $2$ incurs cost at
least $1/2$. So, 
\[
\SC(y_B;B)\ge3.
\]
\end{itemize}
Consequently, in either case, we have that 
\[
\SC(y_B;B)\ge\frac52.
\]
By Claim~\ref{clm:endpoint-opt}, every location has social cost at
least $2$ for profile $A$. Therefore,
\[
\SC(y_A;A)+\SC(y_B;B)
\ge
2+\frac52
=
\frac92.
\]

\medskip
\noindent 
\textbf{Case 2: $1/2<w<1$.}
The two inequalities in~\eqref{eq:lb-left} determine $y_A$ exactly.
Since $d(\tfrac12,w)=w-\frac12$, the first inequality implies that $d(\tfrac12,y_A)\le w-\frac12$.
Because $w-\frac12<1$, there is no truncation in this inequality, and
therefore
\[
\left|y_A-\frac12\right|\le w-\frac12 
\Rightarrow
1-w\le y_A\le w.
\]
On the other hand, since $d(0,w)=w$, the second inequality in~\eqref{eq:lb-left} requires that $d(0,y_A)\ge w$.
Since $y_A\in[1-w,w]\subseteq(0,1)$, we have that $d(0,y_A)=y_A$, and hence $y_A\ge w$. Combining this with $y_A\le w$ gives us that
\[
y_A=w.
\]
Consequently,
\[
\SC(y_A;A)
=
2\left(w-\frac12\right)
+\left(\frac32-w\right)
+1
=
\frac32+w.
\]

We next obtain a lower bound for the cost on profile $B$. Since $d(\tfrac32,w)=\frac32-w$, the first inequality in~\eqref{eq:lb-right} requires $d(\tfrac32,y_B)\ge\frac32-w$. 
Because $3/2-w<1$, this inequality implies that $y_B$ cannot lie
strictly between $w$ and $3-w$. Thus,
\[
y_B\le w
\qquad\text{or}\qquad
y_B\ge3-w.
\]
\begin{itemize}
\item 
If $y_B\le w<1$, the two agents at $2$ each incur cost one. Moreover, by the triangle inequality,
\[
d(0,y_B)+d(\tfrac12,y_B) \ge d(0,\tfrac12) =\frac12.
\]
Hence
\[
\SC(y_B;B)\ge\frac52.
\]
Since $w>1/2$, we have $\frac52\ge3-w$, and thus
\[
\SC(y_B;B)\ge3-w.
\]

\item 
If $y_B\ge3-w>2$, the agents at $0$ and $1/2$ each incur cost one. Furthermore, $y_B-2\ge1-w$, implying that each of the two agents at $2$ incurs cost at least $1-w$. Using also the fact that $w < 1$, we have  
\[
\SC(y_B;B) \ge 2+2(1-w) = 4-2w \ge 3-w. 
\]
\end{itemize}
Hence, in both possible regions for $y_B$,
\[
\SC(y_B;B)\ge3-w.
\]
By combining the bounds for $A$ and $B$, we obtain
\[
\SC(y_A;A)+\SC(y_B;B) \ge \left(\frac32+w\right)+(3-w) = \frac92.
\]

\medskip
\noindent 
\textbf{Case 3: $w\ge1$.}
In this case, $d(0,w)=1$. The second inequality in~\eqref{eq:lb-left} therefore implies $d(0,y_A)=1$, and hence $|y_A|\ge1$.
\begin{itemize}
\item 
If $y_A\le-1$, every agent in $A$ incurs cost one, and hence $\SC(y_A;A)=4$.

\item 
If $1\le y_A\le3/2$, then
\begin{align*}
\SC(y_A;A)
&= 2\left(y_A-\frac12\right) +\left(\frac32-y_A\right) +(2-y_A)
= \frac52.
\end{align*}

\item 
If $3/2\le y_A\le2$, the two agents at $1/2$ each incur cost one,
while
\[
d(\tfrac32,y_A)+d(2,y_A) = \left(y_A-\frac32\right)+(2-y_A) = \frac12.
\]
Hence, $\SC(y_A;A)=5/2$.

\item If $y_A\ge2$, the two agents at $1/2$ contribute two, while the other two agents contribute
at least
\[
d(\tfrac32,y_A)+d(2,y_A) \ge d(\tfrac32,2) =\frac12
\]
by the triangle inequality, giving us that $\SC(y_A;A)\ge 5/2$.
\end{itemize}
Therefore, for every $|y_A|\ge1$,
\[
\SC(y_A;A)\ge\frac52.
\]
By Claim~\ref{clm:endpoint-opt}, every location has social cost at least $2$ for profile $B$. Consequently,
\[
\SC(y_A;A)+\SC(y_B;B) \ge \frac52+2 = \frac92.
\]
The three cases cover every possible value of $w$, proving \eqref{eq:det-sum} and therefore completing the proof.
\end{proof}

\section{Open questions}\label{sec:open}
Several interesting questions remain open. The most immediate one concerns the case $n=4$, where the optimal approximation ratio of universally strategyproof randomized mechanisms currently lies between $9/8$ and $3/2-1/4 = 5/4$. Closing this gap would be a first step toward understanding the power of non-rank-based mechanisms. More generally, determining the optimal approximation ratio of universally strategyproof randomized mechanisms for arbitrary $n$ and studying whether stronger guarantees are possible under truthfulness in expectation remain important directions for future work.

\bigskip

\noindent
\textbf{Declaration of generative AI and AI-assisted technologies.}
During the preparation of this work, the author used OpenAI ChatGPT to assist with proof checking, manuscript organization, and language editing. The author reviewed and edited the content as needed and takes full responsibility for it.

\end{document}